\documentclass[letterpaper, 10 pt, conference]{ieeeconf}  

\IEEEoverridecommandlockouts                              
\usepackage{tikz}
\usepackage{subcaption}
\usepackage{booktabs}

\usepackage{cite}
\usepackage{amsmath,amssymb,amsfonts}
\usepackage{graphicx}
\usepackage{textcomp}
\usepackage{xcolor}
\def\BibTeX{{\rm B\kern-.05em{\sc i\kern-.025em b}\kern-.08em
    T\kern-.1667em\lower.7ex\hbox{E}\kern-.125emX}}

\usepackage{mathrsfs}
\usepackage{dsfont}
\usepackage{algorithm}
\usepackage{makecell}
\usepackage{algpseudocode}

\usepackage{mathptmx}
\DeclareMathAlphabet{\mathcal}{OMS}{cmsy}{m}{n}

\newtheorem{thm}{\bf Theorem }
\newtheorem{prop}[thm]{\bf Proposition}
\newtheorem{deff}[thm]{\bf Definition}
\newtheorem{rem}[thm]{\bf Remark }
\newtheorem{ass}[thm]{\bf Assumption}

\renewcommand{\L}{\mathcal{L}}

\newcommand{\I}{\mathcal{I}}

\newcommand{\C}{\mathcal{C}}

\newcommand{\R}{\mathbb{R}}

\newcommand{\nn}{\mathcal{N}}

\newcommand{\ppb}{\mathsf{P}}
\newcommand{\eeb}{\mathsf{E}}
\newcommand{\fr}{\mathscr{F}}
\newcommand{\Xh}{\widehat{X}}

\newcommand{\add}{\mathsf{ADD}}
\newcommand{\cadd}{\mathsf{CADD}}
\newcommand{\eps}{\varepsilon}

\newcommand{\ra}{\rightarrow}

\newcommand{\set}[1]{\left\{#1\right\}}

\newcommand{\kl}{D_{\operatorname{KL}}}
\newcommand{\arginf}{{\operatorname{arginf}}}

\newcommand{\Qed}{\hfill$\diamond$}

\usepackage{accents}

\newcommand{\tg}{\tilde{g}}
\newcommand{\ttg}{\accentset{\approx}{g}}
 \newcommand{\tf}{\tilde{h}}
\newcommand{\ttf}{\accentset{\approx}{h}}

\usepackage{subcaption}

\title{\LARGE \bf Target Prediction Under Deceptive Switching Strategies via Outlier-Robust Filtering of Partially Observed Incomplete Trajectories}

\author{ Yiming Meng, Dongchang Li, and Melkior Ornik
\thanks{This research was supported by the Office of Naval Research under grant number N00014-23-1-2651.} 
\thanks{
Yiming Meng is with  the
Coordinated Science Laboratory, University of Illinois Urbana-Champaign,
Urbana, IL 61801, USA.
        {\tt\small ymmeng@illinois.edu}.
}
\thanks{
Dongchang Li is with the Department of Applied Mathematics,
University of Waterloo, Waterloo ON N2L 3G1, Canada {\tt\small d235li@uwaterloo.ca}.}
\thanks{Melkior Ornik is with the Department of Aerospace Engineering and the
Coordinated Science Laboratory, University of Illinois Urbana-Champaign,
Urbana, IL 61801, USA.
        {\tt\small mornik@illinois.edu}.}
}

\begin{document}

\maketitle
\thispagestyle{empty}
\pagestyle{empty}

\begin{abstract}
Motivated by a study on deception and counter-deception, this paper addresses the problem of identifying an agent’s target  as it seeks to reach one of two targets in a given environment.  In practice, an agent may initially follow a strategy to aim at one target but decide to switch to another midway. Such a strategy can be deceptive when the counterpart only has access to imperfect observations, which include heavily corrupted sensor noise and possible outliers,  making it difficult to visually identify the agent’s true intent.  To counter deception and identify the true target, we utilize prior knowledge of the agent’s dynamics and the imprecisely observed partial trajectory of the agent’s states to dynamically update the estimation of the posterior probability of whether a deceptive switch has taken place. 
However,   existing methods in the literature have not achieved effective deception identification within a reasonable computation time. We propose a set of   outlier-robust change detection methods to track relevant change-related statistics efficiently, enabling the detection of deceptive strategies in hidden nonlinear dynamics with reasonable computational effort. The performance of the proposed framework is examined for Weapon-Target Assignment (WTA) detection under deceptive strategies, using random simulations in the kinematics model with external forcing.

\end{abstract}


\section{Introduction}

Deception, the act of inducing a false belief in an adversary to achieve a desired objective, is of obvious interest to research in defense \cite{caddell2004deception, daniel1982propositions}, cybersecurity \cite{amin2012cyber,kwon2013security}, social robotics \cite{sharkey2021we,
wagner2011acting}, search and rescue \cite{shim2015benefits}, etc.  In frameworks where the deceptive agent attempts to   reach a particular target, the agent’s trajectory is often used to plant an incorrect belief about its purpose in the adversary \cite{dragan2014analysis,masters2017deceptive,ornik2020measuring, ornik2018deception}.

Motivated by the desire to predict the target of a possibly deceptive agent, this paper presents the problem of quantitatively ascertaining the agent’s target by using a model of the agent dynamics and observing its trajectory under the following constraints: 1) the observations are corrupted by stochastic noise  with the possible presence of outliers, and 2) observations of the complete state are available only at finitely many instants.

Previous efforts in control related deception or counter-deception include the following works. The series of studies in \cite{karabag2019optimal, ornik2020measuring, ornik2018deception} primarily addresses the design of deception strategies against counter-deception measures. 
From the opposite perspective, the work in \cite{ramirez2011goal} addresses the counter-deception problem   from a modeling perspective and proposes a partially observed Markov decision process (POMDP) framework to explore the target prediction problem. Essentially, the model approaches the deception concept    from the probabilistic prior belief of reaching multiple targets and aims to compute the posterior probability of reaching a certain target based on the observation. The computation seeks to minimize the cost between the belief in following some optimal reachability strategy and the actual observation. 
The works \cite{masters2018cost, masters2019cost} further discuss goal recognition based on the framework of path planning and use the ranking of costs associated with a deceptive agent to estimate the probability of reaching one of many targets based on observations. In addition, the recent work \cite{ornik2020measuring} 
  tackled a similar problem of detecting agent deception, although it assumed that the agent chooses a path to the target uniformly at random.  

Previous approaches pave the way toward counterdeceptive target prediction strategies. However, much remains to be improved. The  assumption that an agent chooses its path entirely at random often does not hold in practice.  
even when the observer knows the set of possible decision-making strategies available to the deceptive agent, computing the probability distribution over the agent's paths remains challenging. This is particularly true when observations are partial and may contain outliers. In terms of modeling, prior work has primarily focused on POMDP frameworks, where deception is conceptualized as agents selecting a single control strategy to steer paths in   a way that it seemingly reaches several targets with similar probabilities. This approach requires that the selection of such a  control strategy to be more restrictive.

Considering the aforementioned drawbacks, we model deception and target prediction differently in this paper. To demonstrate the idea, we work on the basic dual-target prediction model where the deceptive agent makes a strategic decision change at an uncertain time. Using corrupted and incomplete observations, we develop a statistical estimation framework that enables effective target prediction. 
Such modeling is well known in hidden Markov models (HMMs), which are generally equipped with nonlinear dynamics. 

From a filtering perspective, we study an effective method of tracking statistics, namely likelihood ratio functions, to compute the conditional probability of reaching one of the two targets. Particularly, we integrate an outlier-robust filter from \cite{hassan2002nonlinear} with the likelihood ratio testing to enhance computational efficiency. This framework simultaneously computes posterior reachability probabilities from finite-horizon partial observations, enabling quantitative estimation of deceptive behavior. 
Additionally, when the observation horizon becomes excessively long, we reuse existing statistics and leverage established quickest change detection (QCD)   to determine an optimal stopping procedure that minimizes the average detection delay following the true change point. 
This approach reduces the average  delay in detecting deceptive changes  while maintaining a low false alarm probability, thus achieving the quickest estimation of target changes with greater statistical confidence.

It is worth noting that QCD has attracted extensive attention over the past decades. The rich literature provides   theoretical guarantees for QCD algorithms in change detection, covering both general signals and HMM-based observation signals \cite{tartakovsky2017asymptotic, fuh2018asymptotic, tartakovsky2005general}. Attempts have also been made to use QCD for target detection in various application contexts \cite{kent2000trail, tartakovsky2002efficient, tartakovsky2004change}, though not specifically for HMMs. Additionally, the QCD problems in nonlinear HMMs constructed from these underlying signals have not been extensively studied and remain poorly understood beyond linearizations. One major difficulty arises from the lack of a robust and accurate likelihood ratio approximation, as mentioned above. It is natural to recall outlier-robust nonlinear filtering techniques that provide recursive algorithms for approximating the conditional density of the state variables. 

Hence, one of the main contributions of this paper is the detailed mathematical formulation of target prediction statistics using an outlier-robust nonlinear filter. This result will also be used to formulate outlier robust filter-enhanced QCD algorithms.
Additionally, we demonstrate through numerical examples how the proposed framework effectively 
enables counter-deception by incorporating more realistic agent modeling assumptions while achieving greater computational efficiency than previous approaches. 
Specifically, we test the results in a case study for Weapon-Target Assignment (WTA) detection under deceptive switching strategies, employing   simulations in the kinematics model. Finally, we discuss the potential of extending the current framework to predict and detect multiple possible targets in real-time, enhancing its suitability for dealing with counter-deception. Future work will build on the insights from this  paper.

\section{Preliminaries for Deceptive Agent Modeling and Target Prediction}\label{sec: pre}
We   use the following basic scenario to introduce the assumptions on the agent model, the concept of deception, and the assumptions on the observations. In this basic scenario, a deceptive agent seeks to move to one of the two preset targets placed in the given environment. The agent may strategically switch targets at an uncertain moment, which can be deceptive when a fair amount of noise contaminates the observations, making it difficult to identify visually. In this scenario,  we assume the role of an observer attempting to detect deception and predict the agent's intended target. This prediction is based on a finite horizon of noisy observations, with the specific objective of determining whether a change point has occurred.

In this section, we provide an overview of the basic heuristics for modeling and the relevant probability measures for target prediction statistics.


\subsection{Standard Modeling of Agent Motion and Targets}
Let $(\Omega, \mathscr{F}, 
\ppb)$ be some probability space, and let $\rho$ denote the   density   of $\ppb$. Suppose the agent has a continuous-time signal $X:=\set{X(t)}_{t\geq 0}$  governed by the following stochastic differential equation: 
\begin{equation}\label{E: sys}
    dX(t)=\left\{\begin{array}{lr} 
F_\alpha(X(t), u(t))dt + \eps B_\alpha dW(t),\;\;t<\nu;\\
F_\beta (X(t), u(t))dt + \eps B_\beta dW(t),\;\;t\geq \nu,   
\end{array}\right. 
\end{equation}
where $W :=\{W(t)\}_{t\geq 0}$ is a   Wiener process; $\eps\in[0,1)$ represent the intensity of the noise term; the quantities $F_j$, $B_j$ for all $j\in\{\alpha,\beta\}$  have
proper dimensions; $u$ denotes the control signal; $\nu$ denotes the moment when the agent switches dynamics.

We suppose the observations are taken at discrete times $t_n:= n\delta_t$ for $n\geq 0$ and some     sampling period $\delta_t$. The observation at index $n$ is
of the form
\begin{equation}\label{E: observation_1}
    Y_n = H(X_n) + V_n,
\end{equation}
where $X_n:=X(n\delta_t)$; $V_n$   is i.i.d. with   distribution $\nn(0, R_n)$ for each $n$,  and $\{V_n\}_{n\geq 0}$ is independent of  $W$; $H$ is the observation channel.   We also introduce the shorthand notation $X_i^n:=\set{X_i, \cdots, X_n}$ for the discrete-time joint states, and $Y_i^n := \set{Y_i, \cdots, Y_n}$ for the observations from $t_i$ to   $t_n$ ($i\leq n)$.

\begin{ass}\label{ass: discrete_time}
    For simplicity in demonstrating the idea of target detection using  outlier-robust filtering, we assume that $\nu$ can only occur at $k\delta_t$ for some random integer $k\geq 0$. 
    
    We adopt the commonly used assumption that the prior knowledge of $\pi_k = \ppb(\nu = t_k)$ follows a geometric distribution, i.e., $\pi_k = d(1-d)^{k-1}$ for some $d \in (0, 1)$.\Qed
\end{ass}

We model the  potential targets $\Gamma_\alpha, \Gamma_\beta$ as  some closed balls centered at some  states $x_{e,\alpha}$ and $x_{e, \beta}$ of the system \eqref{E: sys}. 

\begin{ass}\label{ass: dual_sys} We assume that 
for each $j\in\set{\alpha, \beta}$, we have   knowledge of the control strategy 
$\kappa_j$,  such that each noise-free  system 
\begin{equation}\label{E: sys2}
    dX(t)=
F_j(X(t), \kappa_j(X(t)))dt  
\end{equation}
is exponentially stable w.r.t. $x_{e,j}$ under the state feedback control $u(t)=\kappa_j(x(t))$. 
\Qed
\end{ass}

\begin{rem}
The purpose of introducing the notion of stability is to facilitate the construction of the control law for the deceptive agent  and to guarantee some sufficient conditions that ensure the detection algorithm works. 
Particularly, for system \eqref{E: sys}, exponential stability implies the reach-and-stay property of the noise-free solution  w.r.t. each target when $\eps = 0$ \cite{bertsekas1971minimax}, and ensures probabilistic reachability 
with a probability arbitrarily close to $1$ when $\eps>0$ \cite{10443446}. \Qed

\end{rem}

Note that Assumptions 1 and 2   define a reference distribution of control strategies for deception detection. However, conventional approaches like parameter identification through distribution matching may be ineffective for finite observation horizons, particularly when contaminated by outliers. By leveraging insights into the agent's decision-making heuristics, we propose a QCD-based target prediction method robust to outlier observations, as detailed in Section \ref{sec: qcd_filter}. 
To facilitate the derivation of detection-related statistics, we require an explicit discrete-time state evolution for $X$. By combining Assumptions \ref{ass: discrete_time} and \ref{ass: dual_sys}, this can be expressed as $ X_{n+1} = f_{j,n}(X_n) + \eps B_jW_n, \;\;j\in\{\alpha, \beta\}$. 

Note that, in the equation above, $f_{j,n}$ can be implicitly obtained from a numerical scheme for each $j\in\{\alpha, \beta\}$, where the Euler-Maruyama method is commonly used. Furthermore,  $W_n:= W({(n+1)\tau})-W({n\tau})$ represents the increment of the Wiener process over the interval from $n\tau$ to $(n+1)\tau$. Clearly, for each $j\in\{\alpha, \beta\}$, given that $\eps>0$,$\{\eps B_jW_n\}_{n\geq 0}$ is a Gaussian process; and for each $n\geq 0$, we denote the distribution as $\eps B_jW_n\sim\nn(0, \eps Q_{j, n})$. When $\eps=0$, 
each $\eps B_jW_n$ is a point mass. 


\subsection{Probability Measures for the   Dual-Target Model}\label{sec: intro_qcd}

As the change-point $\nu$ is uncertain to the observer, we aim to detect $\nu$   
based on the observations.  
We now introduce the following frequently-used 
probability measures, which will later be used to determine whether the deceptive decision at $\nu$ has been triggered based on the observation $Y$.


We first introduce two mutually locally absolutely continuous probability laws, 
$\ppb_\infty$ for the normal regime (where no change occurs) and $\ppb_0$ for the abnormal regime (where a change happens at some point), defined on the probability space. Accordingly,  we consider the filtration $\fr_n:=\sigma(Y_0^n)$ as the $\sigma$-algebra generated by the observations, and define the measures restricted to the filtration $\fr_n$ as $\ppb_\infty^{(n)}$ and $\ppb_0^{(n)}$, with their densities denoted as $\rho_\infty(Y_0^n)$ and $\rho_0(Y_0^n)$, respectively. Note that the induced conditional densities $\rho_j(Y_n|Y_0^{n-1})$  for any $j\in\{0, \infty\}$  may depend on $n$, especially in non-i.i.d. cases.   We therefore also write $\rho_{j,n}(Y_n|Y_0^{n-1})$  
when $n$ is emphasized,  and vice versa. Additionally, the post-change conditional probability density $\rho_{0,n}(Y_n|Y_0^{n-1})$ also generally depend on the change point $k$, and we write $\rho_{0,n}^{(k)}(Y_n|Y_0^{n-1})$ accordingly when $k$ is emphasized, and vice versa.

For a fixed $k$, if $\nu = k$, we introduce the change-induced probability measure as $\ppb_k(\cdot) = \ppb(\cdot\;|\nu=k)$, with density  $ 
    \rho_k(Y_0^n) = \rho_\infty(Y_0^{k-1})\cdot \rho_0(Y_k^{n}|Y_0^{k-1}) $
for any $n\geq k$. Using Bayes' rule and emphasizing the potential dependence on the changing point $k$ and the observation period $n$, we can also express $\rho_k(Y_0^n)$ as $\rho_k(Y_0^n) =\left(\prod_{i=0}^{k-1}\rho_{\infty, i}(Y_i|Y_0^{i-1})\right)\cdot\left(\prod_{i=k}^{n}\rho_{0, i}^{(k)}(Y_i|Y_0^{i-1})\right)$.
Recalling that $\pi_k = \ppb(\nu = k)$,  we define another induced (averaging) probability measure $\ppb^\pi(\cdot):=\sum_{k=0}^\infty\pi_k\ppb_k(\cdot)$.   

We denote $\eeb$,  $\eeb_\infty$, $\eeb_0$,   $\eeb_k$, and $\eeb^\pi$ 
as the expectations w.r.t. the probability measures $\ppb$, $\ppb_\infty$, $\ppb_0$,   $\ppb_k$,  and $\ppb^\pi$, 
respectively.

\section{Likelihoods and Procedures for  Change Detection}\label{sec: likelihoods}
There are many ways to identify the agent's path given its deceptive behavior, based on the observation process $Y$. In this section, we   explain how this can be accomplished 
in the context of change estimation.  We start with  an introduction of   likelihood functions.
\subsection{Constant-horizon observation statistics}
Let $p_n:=\ppb(n\geq \nu\;|\;Y_1^n)$ be the \textit{a posteriori} probability that the change occurred before time $t_n$. A probabilistic estimation of whether the deceptive decision at $\nu$ has been triggered,  based on a fixed-horizon observation, is obtained by calculating the statistics for $p_n$. We derive the likelihood ratio $\L_n$ of the hypotheses $\{\nu\leq n\}$ and $\{\nu> n\}$ as follows: 
\begin{small}
    \begin{equation}\label{E: L_n}
\begin{split}
    \L_n &  = \frac{\sum_{k=1}^n\pi_k\prod_{i=1}^{k-1}\rho_{\infty,i}(Y_i|Y_1^{i-1})\prod_{i=k}^{n}\rho_{0,i}^{(k)}(Y_i|Y_1^{i-1})}{\ppb^\pi(\nu>n)\prod_{i=1}^n\rho_{\infty,i}(Y_i|Y_1^{i-1})}\\
    & = \frac{1}{\ppb^\pi(\nu>n)}\sum_{k=1}^n\pi_k\L_n^k,
\end{split}
\end{equation}
\end{small}
\noindent where $\L_n^k=\prod_{i=k}^n \Lambda_i^{(k)}$, 
\begin{equation}\label{E: Lambda}
    \Lambda_i^{(k)}=\frac{ \rho_{0,i}^{(k)}(Y_i|Y_1^{i-1})}{ \rho_{\infty,i}(Y_i|Y_1^{i-1})}, 
\end{equation}
and $\ppb^\pi(\nu>n)$ is the probability
of false alarm (PFA). 

Estimating the probabilistic estimation of $p_n$ (or $\L_n$) necessitates a statistical update of $\Lambda_i^{(k)}$.

\subsection{Brief Introduction to Change Detection Procedures}
The above direct estimation using $\L_n$ determines whether the deceptive decision occurred before a fixed observation stopping time, with probabilistic certainty. 
However, in practice, we are also interested in shortening the observation time and recognizing the deceptive behavior as quickly as possible. 

To improve the efficiency of counterdeception efforts, we introduce the following two metrics \cite{tartakovsky2005general} to guide us in determining the quickest stopping time for observation. 

One reasonable metric of the detection lag is the average
detection delay (ADD), defined as $\add(\tau):=\eeb^\pi(\tau-\nu | \tau \geq \nu)$. It can be shown that $\add(\tau)=\frac{\eeb^\pi(\tau-\nu)^+}{\ppb^\pi(\tau\geq \nu)} 
       =  \frac{1}{\ppb^\pi(\tau\geq \nu)} \sum_{k=1}^\infty \pi_k\ppb_k(\tau\geq k)\eeb_k(\tau-k | \tau\geq k)$. 
Another closely related metric is the conditional average
detection delay (CADD), defined as $\cadd(\tau):=\sup_{k\geq 1}\eeb_k(\tau-k | \tau \geq k)$, 
which captures the worst case scenario. 

Constrained by the need to maintain a lower probability of false alarm  $a\in(0,1)$, we work on the set $\C(a):=\set{\tau: \ppb^\pi(\tau<\nu)\leq a}$
and determine the optimal stopping strategy by either $\arginf_{\tau\in\C(a)}\add(\tau)$ or $\arginf_{\tau\in\C(a)}\cadd(\tau)$.

There is a rich literature proving that, under mild conditions, the Shiryaev stopping rule 
\begin{equation}\label{E: Shiryaev}
    \tau_s(B_a)=\inf\{n\geq 1: \L_n\geq B_a\}, \;B_a=\frac{1-a}{a},
\end{equation}
can asymptotically solve the optimization problem for $\arginf_{\tau\in\C(a)}\add(\tau)$ as $a\ra 0$ \cite{tartakovsky2005general}. 

Similarly, by defining $Z_i^{(k)}:=\log (\Lambda_i^{(k)})$ and $T_n = \max_{1\leq k\leq n} \sum_{i=k}^n Z_i^{(k)}$, 
we use the cumulative sum (CUSUM) stopping procedure 
\begin{equation}\label{E: cusum}
    \tau_c=\inf\set{n\geq 1: T_n\geq c}
\end{equation}
to aymptotically solve $\arginf_{\tau\in\C(a)}\cadd(\tau)$ as $c\ra \infty$.  Note that the CUSUM rule is designed to check the worst-case risk scenario and does not require prior knowledge of $\pi_k$, making it more flexible than the Shiryaev rule for predicting deceptive behavior even when $\nu$ is unknown. 

\section{Change Detection Using Outlier-Robust Filtering for the Dual-Target Mode}\label{sec: qcd_filter}
In this section, we construct the change detection statistics using  outlier-robust nonlinear filter. The key step is to compute the quantity $\Lambda_i^{(k)}$ as defined in  \eqref{E: Lambda}. We first derive the formula for the likelihood function of general hidden Markov models (HMM) when outliers may exist and an outlier-robust nonlinear filter is required. 
We then integrate the filtering strategy to demonstrate the filter-based approximation of the likelihood function.

\subsection{Likelihood function for HMMs with the appearance of outlier indicators}\label{sec: derivation_likelihood}
At this stage, we do not distinguish between pre-change and post-change probability measures or densities to concisely conduct derivation of likelihood functions using the outlier-robust filtering method in \cite{chughtai2022outlier}.

We first introduce an indicator vector $\I_i\in\R^m$, where for each independent sensor at dimension $l$ at time $t_i$, 
we let 
\begin{equation}\label{E: indicator}
   \I_{i,l}=\left\{\begin{array}{lr} 
 \varsigma>0, \;\text{if an outlier occurs}\\
 1,\;\text{otherwise}.
\end{array}\right. 
\end{equation}
By assigning the probability of no outlier in the $l$-th observation as $\theta_{i,l}\in[0, 1]$ for each instant $i$, the density of $\I_i$ can be explicitly expressed as  $\rho(\I_{i,l})   =\prod_{l=1}^m\rho(\I_{i,l})  = \prod_{l=1}^m\left[(1-\theta_{i,l})\delta(\I_{i,l}-\varsigma)+\theta_{i,l}\delta(\I_{i,l}-1)\right]$.

We   assume that observations are obtained from independent sensors, and consequently, we model the outliers independently for each observation dimension. We also assume that $\I_i$ and $X_i$ are  
  independent. 

We now derive the outlier-robust likelihood function for HMMs. By \eqref{E: sys} and Assumption \ref{ass: dual_sys}, the evolution of the state process $\{X_n\}_{n\geq 0}$
from time $t_i$ to time $t_{i+1}$ satisfies Markov properties, and 
is determined by the transition probability   $\ppb(X_{i+1}\in A |X_i)=\int_A \rho(x|X_i)dx$, $\forall i$. 
The observations $\{Y_n\}_{n\geq 0}$ should satisfy $\ppb(Y_{i}\in B| X_0^{i}, \I_i, Y_0^{i-1})=\int_B \rho (y |X_i, \I_i)dy$, $\forall i$. 

To distinguish the transition along the state and the observation, we denote $g_i(X_{i-1}, X_i):= \rho(X_i |X_{i-1})$
as transition probability densities, and denote $h_i(Y_i|X_i, \I_i):= \rho(Y_i|X_i, \I_i)$
as the observation likelihood.

For simplicity, we let $\widehat{X}_i$ denote the joint variable $(X_i, \I_i)$ for each $i$. 
Then, the joint density $\rho(X_i, \I_i, Y_0^i)=\rho(\Xh_i, Y_0^i)$ is such that
\begin{small}
    \begin{equation}\label{E: joint_density_XY}
\begin{split}
       & \rho(\Xh_i, Y_0^i) \\
       = &\int \rho(\Xh_{i-1}, \Xh_i, Y_0^i)d\Xh_{i-1}\\
       = & \int \rho(\Xh_{i-1},  Y_0^{i-1})\rho(\Xh_i|\Xh_{i-1})h_i(Y_i|\Xh_i)d\Xh_{i-1}\\
         = & \int \rho(\Xh_{i-1},  Y_0^{i-1})g_i(X_{i-1}, X_i)h_i(Y_i|\Xh_i)\rho(\I_i)d\Xh_{i-1}. 
\end{split}
\end{equation}
\end{small}
\noindent 
By averaging out the $\Xh_i$, we obtain 
\begin{small}
     \begin{equation}\label{E: Y_0^n}
\begin{split}
       & \rho( Y_0^i) \\
       = &\iint \rho(\Xh_{i-1},  Y_0^{i-1})g_i(X_{i-1}, X_i)h_i(Y_i|\Xh_i)\rho(\I_i)d\Xh_{i-1}d\Xh_i. 
\end{split}
\end{equation}
\end{small}

\noindent Then, it is clear that
\begin{small}
   \begin{equation}\label{E: formula_1}
\begin{split}
        & \rho(Y_i|Y_0^{i-1})=\frac{\rho( Y_0^i)}{\rho( Y_0^{i-1})}\\
    = & \frac{1}{\rho( Y_0^{i-1})}\iint  \rho(\Xh_{i-1},  Y_0^{i-1})g_i(X_{i-1}, X_i) h_i(Y_i|\Xh_i)\rho(\I_i)d\Xh_{i-1}d\Xh_i\\
   = & \iint  \underbrace{g_i(X_{i-1}, X_i)}_\text{state transition}\cdot\underbrace{h_i(Y_i|\Xh_i)}_\text{observation}\cdot\underbrace{\rho(\Xh_{i-1}|Y_0^{i-1})}_\text{outlier-robust filtering}\cdot\rho(\I_i)d\Xh_{i-1}d\Xh_i.
\end{split}
\end{equation} 
\end{small}

Now, we provide the explicit form of  $\Lambda_i^{(k)}$ based on the derivation in Eq. \eqref{E: formula_1}, with distinguished pre/post-change probability densities. In this case, for each $k$, we write $ g_i(X_{i-1}, X_i)=\tg_i(X_{i-1}, X_i)$ if $i<k$ and $g_i(X_{i-1}, X_i)=\ttg_i^{(k)}(X_{i-1}, X_i)$ otherwise. Similarly, $h_i(Y_i|\Xh_i) = \tf_i(Y_i|\Xh_i)$ if $i<k$, and $h_i(Y_i|\Xh_i) = \ttf_i(Y_i|\Xh_i)$ otherwise.
Then, for $i\geq k$, the likelihood ratio $\Lambda_i^{(k)}$ can be explicitly written as   \eqref{E: likelihood_formula} below. 
\begin{figure*}[!t]
\normalsize

\begin{align}
  \Lambda_i^{(k)}  =\frac{ \rho_{0,i}^{(k)}(Y_i|Y_1^{i-1})}{ \rho_{\infty,i}(Y_i|Y_1^{i-1})}& =  \frac{\iint \ttg_i^{(k)}(X_{i-1}, X_i)\ttf_i(Y_i|\Xh_i) \rho_k(\Xh_{i-1}|Y_0^{i-1})\rho(\I_i)d\Xh_{i-1}d\Xh_i}{\iint  \tg_i (X_{i-1}, X_i)\tf_i(Y_i|\Xh_i)\rho_\infty(\Xh_{i-1}|Y_0^{i-1})\rho(\I_i)d\Xh_{i-1}d\Xh_i} \label{E: likelihood_formula}   
\end{align}

\begin{equation}\label{E: P_ik}
     P_{i,k}^-= \left\{\begin{array}{lr} 
\int (f_{\alpha,i}-m_i^-)(f_{\alpha,i}-m_i^-)^\mathsf{T}q(X_{i-1})dX_{i-1} + Q_{i-1}, \;i<k\\
\int (f_{\beta,i}-m_i^-)(f_{\beta,i}-m_i^-)^\mathsf{T}q(X_{i-1})dX_{i-1} + Q_{i-1},   \;i\geq k. 
\end{array}\right. 
\end{equation}

 \begin{equation}
\begin{split}
    \widehat{\Lambda}_i^{(k)} &     =  \frac{\iint \ttg_i^{(k)}(X_{i-1}, X_i)\ttf_i(Y_i|\Xh_i) q_k(X_{i-1})q_k(\I_{i-1}) \rho(\I_i)d\Xh_{i-1}d\Xh_i}{\iint  \tg_i (X_{i-1}, X_i)\tf_i(Y_i|\Xh_i) q_\infty(X_{i-1})q_\infty(\I_{i-1})\rho(\I_i)d\Xh_{i-1}d\Xh_i},\;\;i\geq k.
\end{split}
\label{E: likelihood_formula_2}
\end{equation}

\hrulefill
\vspace*{4pt}
\end{figure*}


\subsection{Construction of the change detection statistics
using outlier-robust filter}


We aim to provide explicit expressions for the terms in the integrands of \eqref{E: likelihood_formula} 
to demonstrate how they can be inferred by an outlier-robust filter. 

\subsubsection{Measurement Likelihood}
The measurement likelihood conditioned on the
current state $X_i$ and the indicator $\I_i$, independent of all the
historical observations $Y_0^{i-1}$, is proposed to follow a Gaussian
distribution
\begin{small}
    \begin{equation}\label{E: measurement}
    \begin{split}
       & h_i(Y_i|\widehat{X}_i)\\
       = & \nn(Y_i | H(X_i), \Sigma_i^{-1})\\
       = & \frac{1}{\sqrt{(2\pi)^m|\Sigma_i^{-1}|}}\exp\left\{-\frac{1}{2}(Y_i-H(X_i))^\mathsf{T}\Sigma_i(Y_i-H(X_i))\right\}\\
       = & \prod_{l=1}^m\frac{1}{\sqrt{2\pi R_i^{(ll)}/\I_{i,l}}}\exp\left\{-\frac{(Y_i^{(l)}-H^{(l)}(X_i))^2}{2R_i^{(ll)}}\I_{i,l} \right\},
    \end{split}
\end{equation}
\end{small}

\noindent where $\Sigma_i:=R_i^{-1}\operatorname{diag}(\I_i)$ \cite{chughtai2022outlier}.

\subsubsection{Variational Bayesian Inference for Outlier Robust   Filters}

In the conventional filtering problem, the conditional probability density $\rho(X_{i-1}|Y_0^{i-1})$ is recursively updated by the prediction procedure $\rho(X_{i}|Y_0^{i-1})\varpropto g_i(X_{i-1}, X_i)\rho(X_{i-1}|Y_0^{i-1})$ and the filtering procedure $\rho(X_{i}|Y_0^{i})\varpropto h_i(Y_i|X_i)\rho(X_{i}|Y_0^{i-1})$. For nonlinear systems, $\rho(X_{i-1}|Y_0^{i-1})$ can be approximated by a collection of particles with discrete masses and updated by the standard prediction and filtering procedures \cite{baxendale2025quickest}. However, particle filters
suffer from issues of computational efficiency and scalability. 

To reduce the computational complexity involved in sequential approximating $\rho(\Xh_{i-1}|Y_0^{i-1})$ or $\rho(\Xh_{i}|Y_0^{i-1})$, we  resort to the standard Variational Bayes (VB) method, where the joint posterior is approximated as a product of marginal distributions
    \begin{equation}\label{E: VB}
   \rho_k(\Xh_{i}|Y_0^{i}) \approx q_k(X_i)q_k(\I_i), \;k\in\{1, 2, \cdots, \infty\}. 
\end{equation}
  The VB approximation aims to minimize the Kullback-Leibler (KL) divergence  \cite{lovric2011international} 
   between the r.h.s. and l.h.s. of \eqref{E: VB}. Accordingly, the terms in the above product approximation can be updated in a manner described as $q(X_i)\varpropto e^{\eeb_{q(\I_i)}[\ln(\rho(\Xh_{i}|Y_0^{i})]}$ and $q(\I_i)\varpropto e^{\eeb_{q(X_i)}[\ln(\rho(\Xh_{i}|Y_0^{i})]}$, where $\eeb_{q(\cdot)}$ represents the expectation operator with respect to the distribution $q(\cdot)$.

For tractability, we integrate general Gaussian filtering results into the VB framework and extend the method described in \cite{chughtai2022outlier}. The updates for $q_k(X_i)$ and $q_k(I_i)$ for any $k\in\set{1, 2, \cdots, \infty}$
used in \cite{chughtai2022outlier} are summarized as follows.

For $q_k(X_i)$, it is  approximated by a Gaussian distribution, i.e.,  $q_k(X_i)\approx\mathcal{N}(X_i | m_{i,k}^+, P_{i,k}^+)$, and the (observation-dependent) mean $m_{i,k}^+$ and covariance $P_{i,k}^+$ are sequentially updated by the following prediction and filtering procedure. The $m_{0,k}^+$ and $P_{0,k}^+$ are initialized with a known distribution, which is a Dirac measure at $X_0$ if the initial condition of the system is known to the observer. For $i\geq 1$, we approximate the predictive distribution as $\rho_k(X_{i}|Y_0^{i-1}) \approx \nn(X_i|m_{i,k}^-, P_{i,k}^-)$, where 
\begin{equation}
    m_{i,k}^-= \left\{\begin{array}{lr} 
\int f_{\alpha,i}(X_{i-1})q_k(X_{i-1})dX_{i-1}, \;i<k,\\
\int f_{\beta,i}(X_{i-1})q_k(X_{i-1})dX_{i-1},   \;i\geq k. 
\end{array}\right. 
\end{equation}
and $P_{i,k}^-$ is given in \eqref{E: P_ik}.

The parameters at the filtering stage are updated by $m_{i,k}^+ = m_{i,k}^- + K_{i,k}\left(Y_i-\mu_{i,k}\right)$ and $P_{i,k}^+=P_{i,k}^- + C_{i,k}K_{i,k}^\mathsf{T}$, 
where
\begin{small}
    \begin{align}
    & K_{i,k} = C_{i,k}(V_{i,k}^{-1}-V_{i,k}^{-1}(I+U_{i,k}V_{i,k}^{-1})^{-1}U_{i,k}V_{i,k}^{-1});\\
    & \mu_{i,k} = \int H(X_i)\rho_k(X_i|Y_0^{i-1})dX_i;\\
    & U_{i,k} = \int (H(X_i)-\mu_i) (H(X_i)-\mu_{i,k})^\mathsf{T}\rho_k(X_i|Y_0^{i-1})dX_i;\\
    & C_{i,k} = \int (X_i-m_{i,k}^-)(H(X_i)-\mu_{i,k})^\mathsf{T}\rho_k(X_i|Y_0^{i-1})dX_i,
\end{align}
\end{small}

\noindent and $V_{i,k}^{-1}=R_i^{-1}\left(\operatorname{diag}(\eeb_{q_k(\I_i)}(\I_i))\right)$. 

For $q_k(\I_i)$,  based on the VB approximation $q_k(\I_i)\varpropto e^{\eeb_{q_k(X_i)}[\ln(\rho(\Xh_{i}|Y_0^{i})]}$, the explicit formula is given as $q_k(\I_i) = \prod_{l=1}^m(1- \Phi_{i,l}^{(k)})\delta(\I_{i,l}-\varsigma) + \Phi_{i,l}^{(k)}\delta(\I_{i,l}-1)$, 
where  $\Phi_{i,l}^{(k)} = \frac{1}{1 + \sqrt{\varsigma}(\frac{1}{\theta_{i,l}}-1)\exp\left(\frac{W_{i,k}^{(ll)}}{2R_i^{(ll)}}(1-\varsigma)\right)}$
and $W_{i,k}^{(ll)} = \eeb_{q_k(X_i)}(Y_i^{(l)}-H^{(l)}(X_i))^2$. 

As a quick summary, to use the approximation $ \rho_k(\Xh_{i}|Y_0^{i})$ for any fixed $i$ and for any $k\in\set{1, 2, \cdots, \infty}$, the key is to sequentially update the (observation-dependent) pairs $(m_{\iota,k}^+, P_{\iota,k}^+)$ and $(m_{\iota+1,k}^-, P_{\iota+1,k}^-)$ for any $\iota\in\set{0, 1, \cdots, i}$. Real observation information $Y_\iota$ is only injected into $(m_{\iota,k}^+, P_{\iota,k}^+)$ at each filtering stage for $\iota\in\set{0, 1, \cdots, i}$ and will cumulatively contribute to the eventual $\rho_k(\Xh_{i}|Y_0^{i})$.

\subsubsection{Outlier-Robust Filters Induced Likelihood Ratio Function}
Combining \eqref{E: measurement} and \eqref{E: VB}, we obtain the explicit formula for the outlier-robust filters  induced likelihood ratio $\widehat{\Lambda}_i^{(k)}$, as shown in  \eqref{E: likelihood_formula_2}. This formula simply replaces the corresponding conditional probability density $\rho_k(\Xh_{i-1}|Y_0^{i-1})$ with the approximators $q_k(X_{i-1})q_k(\I_{i-1})$ for any $1\leq k\leq i$. The same applies to $\rho_\infty(\Xh_{i-1}|Y_0^{i-1})$. 

Note that, for special case where there are no outliers, $\eeb_{q_k(\I_{i,j})}(\I_{i,j})=1$  for all $j$-th entry, and the update for $(m_{i,k}^+, P_{i,k}^+)$ becomes the standard Gaussian
filtering problem. In this case, the $\rho(\I_i)$ in \eqref{E: likelihood_formula_2} can also be removed. 

\section{Quantitative Change Estimation Using
Outlier-Robust Filters}
 We discuss how the outlier-robust filter informs change estimation within the framework introduced in Section \ref{sec: likelihoods}. Due to page limitations, we only provide a sketched version of the proofs. 
 
 Let $\widehat{\L}_n := \frac{1}{\ppb^\pi(\nu>n)}\sum_{k=1}^n\pi_k \widehat{\L}_n^k$, 
where $ \widehat{\L}_n^k=\prod_{i=k}^n  \widehat{\Lambda}_i^{(k)}$. Then, we have the following approximation result.

\begin{prop}
    For each $n$ and each realization $y_0^n$ of the observation process $Y_0^n$, there exists a constant $C$  such that $|\widehat{\L}_n -\L_n|\leq C\cdot \kl(\rho_k(\Xh_{i-1}|y_0^{i-1})||q_k(X_{i-1})q_i(\I_{i-1})) + \mathcal{O}(\epsilon)$, where $\widehat{\L}_n$ and $\L_n$ are real-valued, and $\mathcal{O}(\eps)\ra 0$ as $\eps\ra 0$. 
\end{prop}
\begin{proof}
        Note that for each realization, each random probability measure in \eqref{E: likelihood_formula} and \eqref{E: likelihood_formula_2}  becomes a probability. Let   $\tau^*_\eps:=\inf\set{t\geq 0: \|X_t\|\geq \eps^{-z}}$ for some $z\in(0, 1)$. Then, due to the exponential stability property based on Assumption \ref{ass: dual_sys}, it can be shown that   $\mathds{1}_{\{t_i<\tau^*_\eps\}} \ra 1$ as $\eps\ra 0$.  For each $k\in\set{1, 2, \cdots, \infty}$, let $\widetilde{\rho}_{i,k,y}(\Xh_{i-1}, \Xh_i):=g_i(X_{i-1}, X_i)h_i(y_i|\Xh_i)\rho_k(\Xh_{i-1}|y_0^{i-1})\rho(\I_i)$ denote the joint density, and similarly $\widetilde{q}_{i,k,y}(\Xh_{i-1}, \Xh_i):=g_i(X_{i-1}, X_i)h_i(y_i|\Xh_i)q_k(X_{i-1})q_k(\I_{i-1})\rho(\I_i)$, where $y$   represents $y_0^i$. 
        
        Then, $\widetilde{q}_{i,k,y}(\Xh_{i-1}, \Xh_i)\leq\widetilde{q}_{i,k,y}(\Xh_{i-1}, \Xh_i)\mathds{1}_{\{t_i<\tau^*_\eps\}\cap \{t_{i-1}<\tau^*_\eps\}}+\widetilde{q}_{i,k,y}(\Xh_{i-1}, \Xh_i)\mathds{1}_{\{t_i\geq\tau^*_\eps\}\cup \{t_{i-1}\geq\tau^*_\eps\}}$, where the first term is uniformly continuous to $q_k(X_{i-1})q_k(\I_{i-1})$ given the continuity of $g_i$ in $X_{i-1}$ and the boundedness on $D_\eps:=\set{x: \|x\|<\eps^{-z}}$, and the second term is of $\mathcal{O}(\eps)$. Based on this property, by separating the integrand in   the definition of $\kl$ as above, one can show that $\kl(\widetilde{q}_{i,k,y}||\widetilde{\rho}_{i,k,y})$ is uniformly continuous to $\kl(\rho_k(\Xh_{i-1}|y_0^{i-1})||q_k(X_{i-1})q_i(\I_{i-1}))$. 
        Combining this with   the well known  Pinsker's inequality \cite{csiszar2011information}, which bounds the total variation norm of two distribution by the KL divergence,  we have $|\iint_A  \widetilde{q}_{i,k,y}(\Xh_{i-1}, \Xh_i)d\Xh_{i-1}d\Xh_i - \iint_A\widetilde{\rho}_{i,k,y}(\Xh_{i-1}, \Xh_i) d\Xh_{i-1}d\Xh_i|$ is uniformly continuous to $\kl(\rho_k(\Xh_{i-1}|y_0^{i-1})||q_k(X_{i-1})q_i(\I_{i-1})) + \mathcal{O}(\epsilon)$ for any measurable event $A$ of $(\Xh_{i-1},\Xh_i)$, which indicates that there exists some $\widehat{C}>0$ that $|\widehat{\Lambda}_i^{(k)}-\Lambda_i^{(k)}|\leq \widehat{C}\kl(\rho_k(\Xh_{i-1}|y_0^{i-1})||q_k(X_{i-1})q_i(\I_{i-1})) + \mathcal{O}(\epsilon)$ for each $k\leq i$. The conclusion therefore follows immediately from the definitions of $\widehat{\L}_n$ and $\L_n$.
\end{proof}

The convergence of $\kl(\rho_k(\Xh_{i-1}|y_0^{i-1})||q_k(X_{i-1})q_i(\I_{i-1})) $ 
is guaranteed  by \cite{chughtai2022outlier}. We continue to review the optimal stopping procedure for QCD algorithms that utilize outlier-robust filters, as briefly discussed in Section \ref{sec: likelihoods}.

For models where $\Lambda_i^{(k)}$ is independent of $k$,  the Shiryaev stopping rule has been proven to have the following property \cite{tartakovsky2005general}. Assuming $\frac{1}{n}\sum_{i=k}^n\log(\Lambda_i)\ra \phi$ as $n\ra\infty$ almost surely in $\ppb_k$ for every $k$, 
\begin{small}
    \begin{equation}\label{E: aysmp_optimality}
    \inf_{\tau\in\mathcal{C}(a)}\add(\tau)\sim\add(\tau_s(B_a))\sim\frac{|\log a|}{\phi+|\log(1-d)|}\;\;\text{as}\;a\ra0,
\end{equation}
\end{small}

\noindent recalling that $d\in(0, 1)$ is such that $\pi_i=d(1-d)^{k-1}$. In practice, it has been shown that Assumption \ref{ass: dual_sys} can guarantee the almost-sure convergence of $\frac{1}{n}\sum_{i=k}^n\log(\Lambda_i)$ \cite{fuh2018asymptotic}. 

\begin{rem}
   Note that in \eqref{E: aysmp_optimality}, optimality can  be achieved asymptotically as $a\ra 0$ for non-i.i.d. observations. 

   The intuition behind the proof of \eqref{E: aysmp_optimality} is to first establish the lower bound of $\add(\tau)$ for any $\tau\in\mathcal{C}(a)$, which is achieved by applying Chebyshev's inequality  $\add(\tau)=\frac{\eeb^\pi[(\tau-\nu)^+]}{\ppb^\pi(\tau\geq\nu)} 
           \geq  \frac{(1-\sigma)|\log a|}{(\phi+|\log(1-d)|)}\left[1- \frac{\gamma_{\sigma, a}(\tau)}{\ppb^\pi(\tau\geq\nu)} \right]$ 
for any $\sigma\in(0, 1)$, where $\gamma_{\sigma, a}(\tau)=\ppb^\pi\set{\nu\leq \tau<\nu+(1-\sigma)\frac{|\log a|}{\phi+|\log(1-d)|}}$ can be shown to converge to $0$ for any $\sigma\in(0, 1)$ as $a\ra 0$. This fact indicates that, even when aiming to reduce the average lag of change detection, the requirement to accommodate a small probability of false alarms causes the detection procedure to place greater emphasis on the tail for time instants longer than $\nu+\frac{|\log a|}{\phi+|\log(1-d)|}$. On the other hand, one can show that $\eeb_k[(\tau_s(B_a)-k)^+]\leq \eeb_k[\eta(k)\mathds{1}_{\tau_s(B_a)\geq k}]\leq \eeb_k[\eta(k)]$, where 
$\eta(k)=\inf\set{n\geq 1: \sum_{i=k}^{k+n-1}\Lambda_i+n\log(1-d)\geq \log(B_a)}$ and $\eeb_k[\eta(k)/\log(B_a)]\ra\frac{1}{\phi+\log(1-d)}$ as $a\ra 0$. This fact indicates that the tail effect of the distribution of $\pi_k$ does not distort the optimality, and it facilitates establishing the upper bound for $\add(\tau_s(B_a))$ as  $a\ra 0$. \Qed 
\end{rem}

\begin{figure*}[!t]
\normalsize
\begin{subfigure}[b]{0.45\textwidth}
        \includegraphics[scale=0.38]{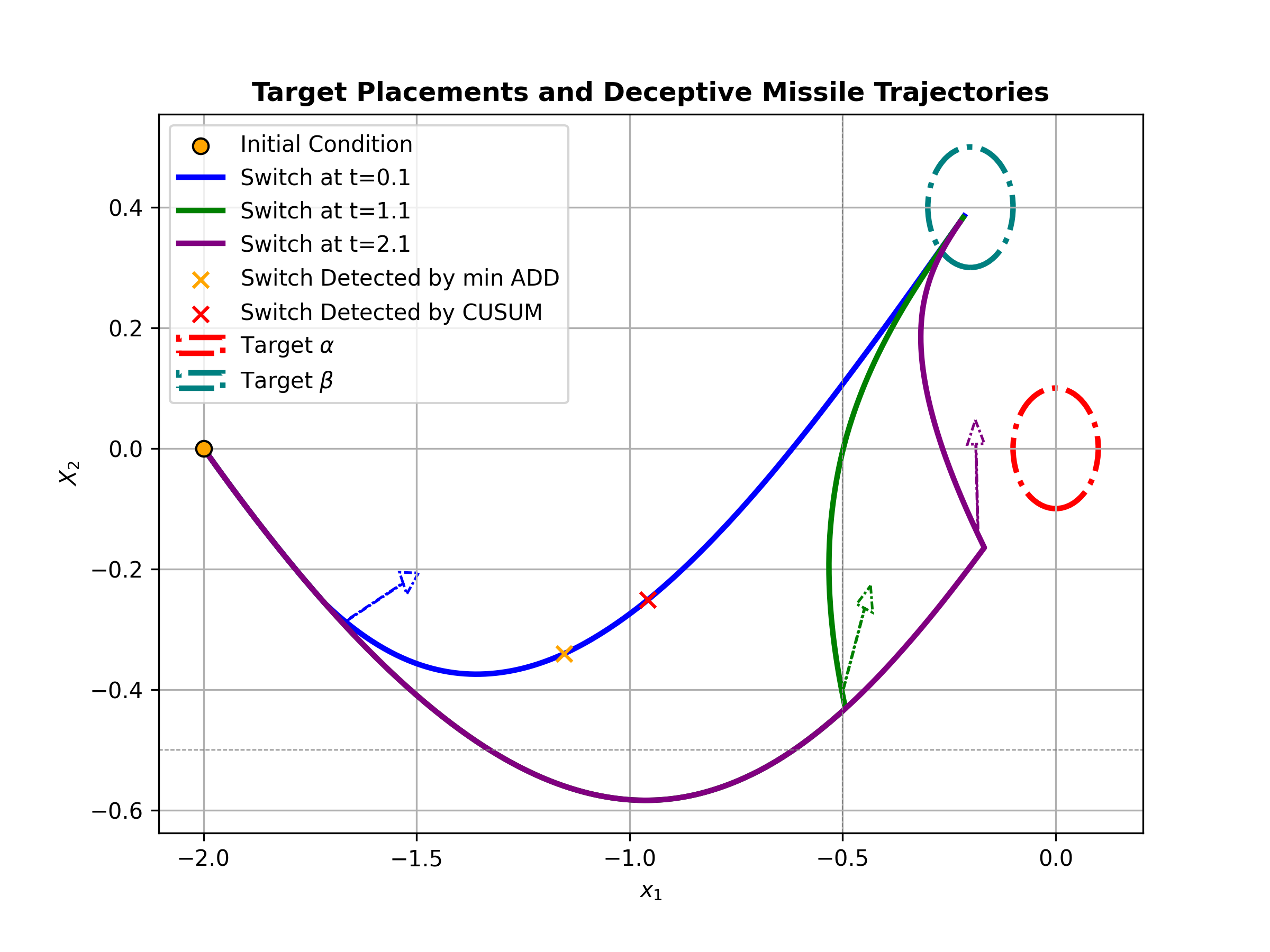}
        \caption{Missile Trajectories Under  Deceptive Switching Strategies}
        \label{fig:subfig1}
    \end{subfigure}
\begin{subfigure}[b]{0.45\textwidth}
        \includegraphics[scale=0.38]{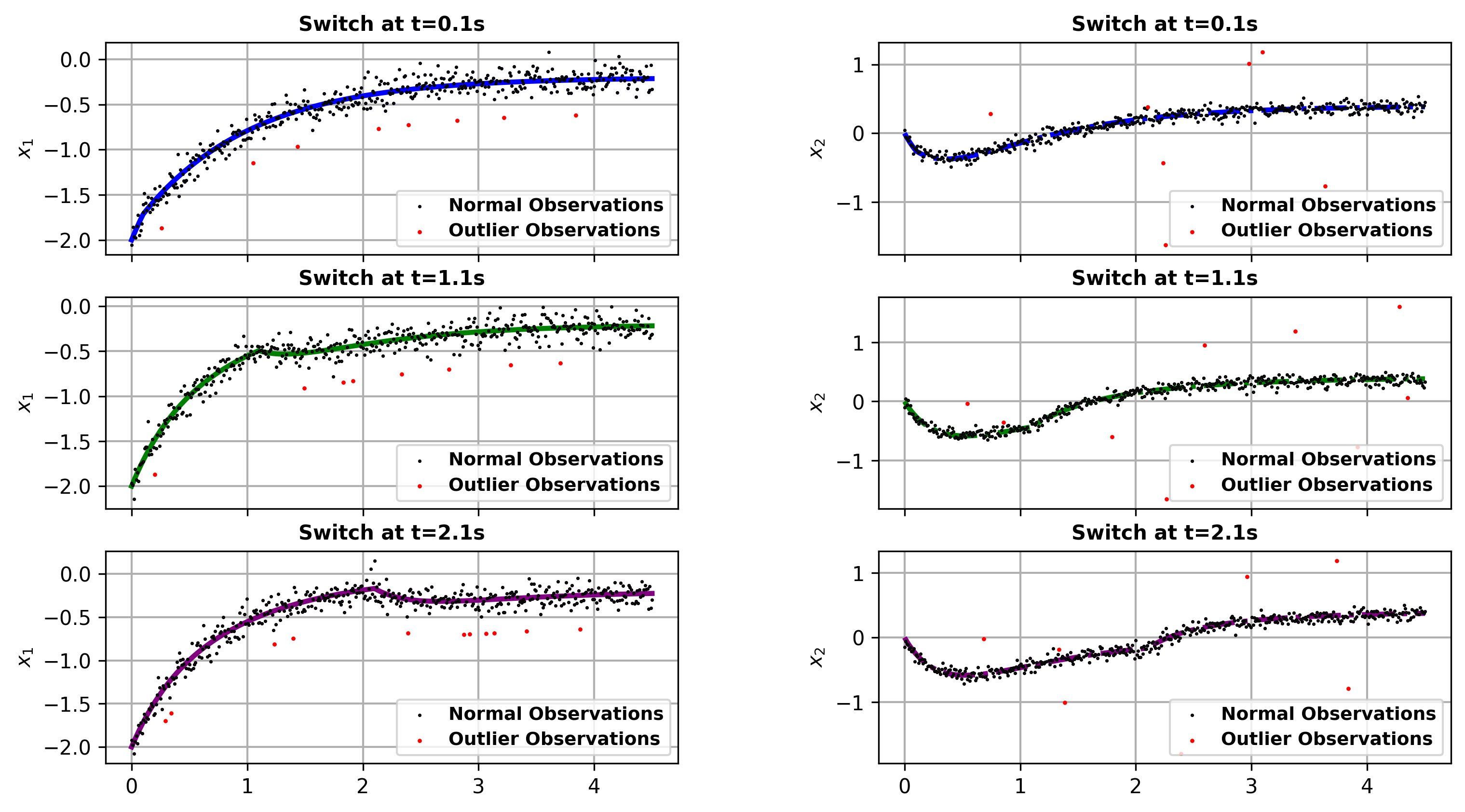}
        \caption{A Realization of Observation Signals}
        \label{fig:subfig2}
    \end{subfigure}    
    
    \begin{subfigure}[b]{0.42\textwidth}
        \includegraphics[scale=0.35]{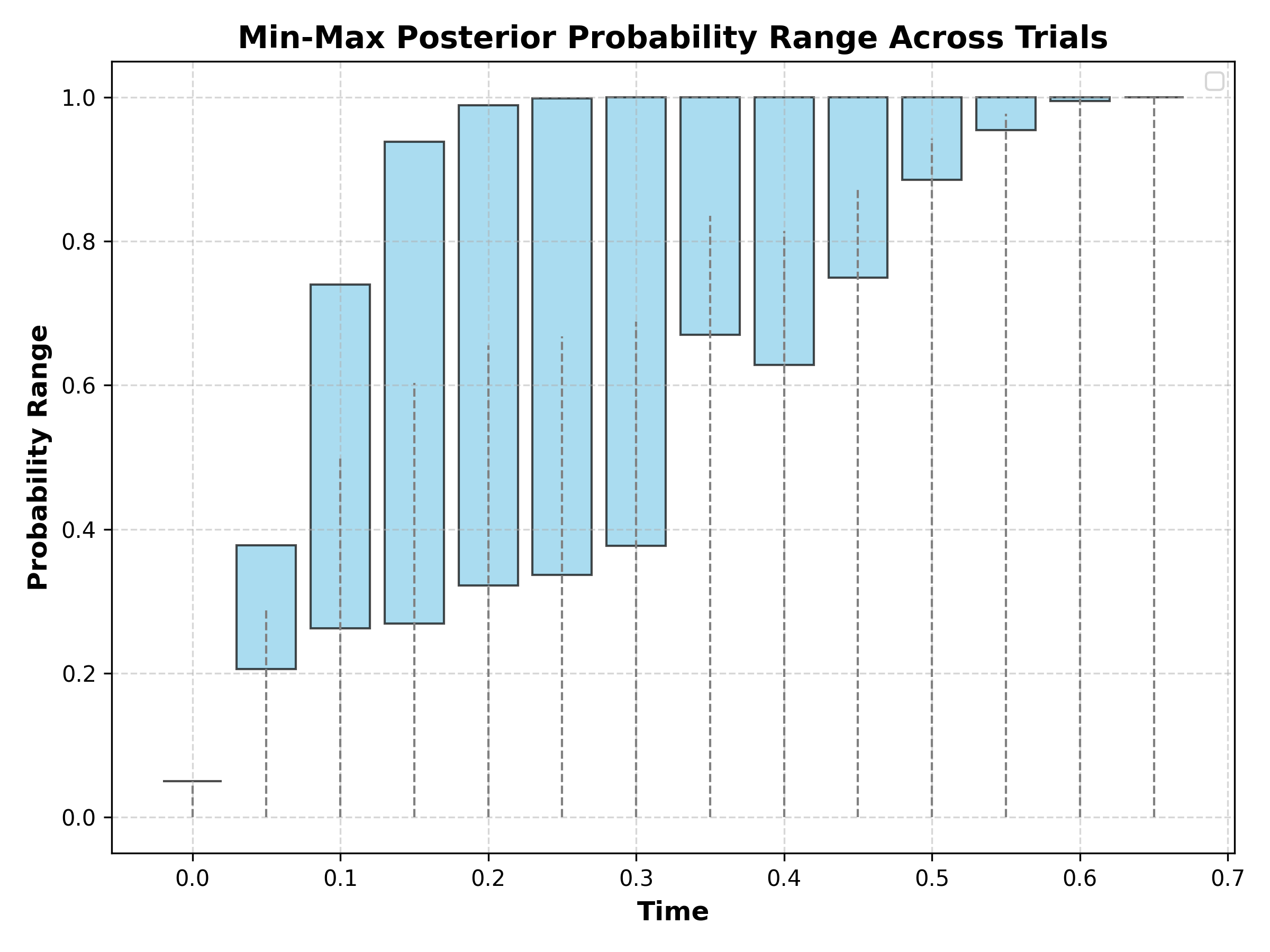}
        \caption{Posterior probability estimation for the  switching strategy using $\widehat{\L}_n$}
        \label{fig:subfig3}
    \end{subfigure}    \qquad\qquad\qquad\quad
    \begin{subfigure}[b]{0.42\textwidth}
        \includegraphics[scale=0.35]{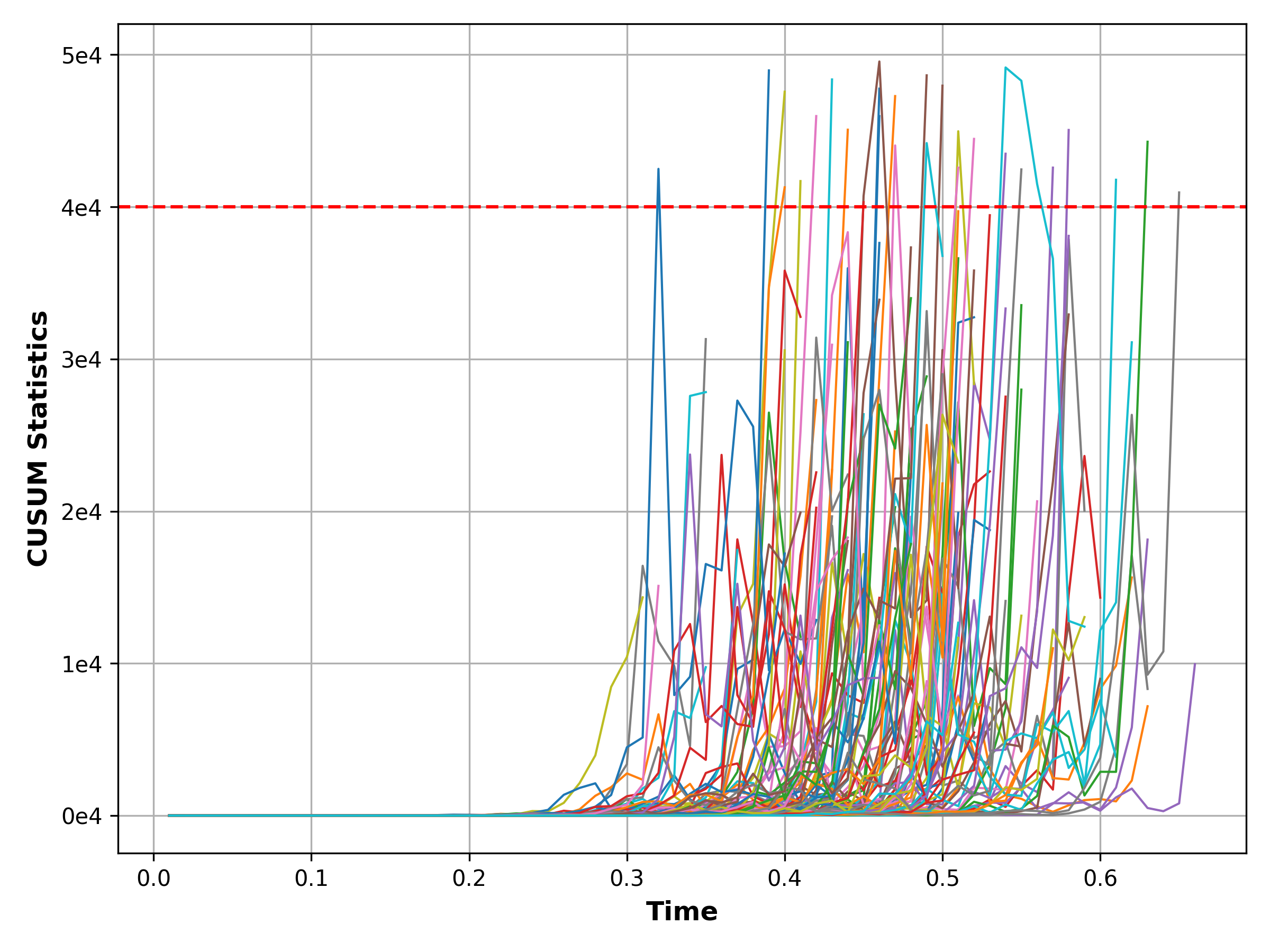}
        \caption{CUSUM statistics for minimizing CADD}
        \label{fig:subfig4}
    \end{subfigure}    
    \caption{Missile Trajectories and Corresponding Observation Signals Under Different Deceptive Switching Strategies.}
    \label{fig:mainfig}
\hrulefill
\vspace*{4pt}
\end{figure*}

In practice, when using the Shiryaev stopping rule with the outlier-robust filter, we need to set a small $a$, and then replace $\Lambda_i$ with $\widehat{\Lambda}_i$ to   track the statistic following \eqref{E: Shiryaev}. Note that $\phi$ indicates the rate at which the difference between the pre- and post-change distributions accumulates on the log scale after the agent has made the deception decision. The cumulative error of using $\widehat{\Lambda}_i$ naturally depends on how different the post-change signal is from the unchanged signal.

Similarly, to detect the worst-case average delay of change detection, the CUSUM stopping procedure $\tau_c$ in \eqref{E: cusum} is proven to be asymptotically optimal as $a\ra0$ \cite{fuh2003sprt}.  One can update  $T_n$ using the $\widehat{\Lambda}_n$ in practice for approximation. 

\section{Case Study}
In this section, we introduce the scenario in which we aim to predict the target of an attacking missile, commonly modeled using a normalized unicycle model \cite{chen2019nonlinear, merkulov2024reinforcement}. For simplicity in demonstrating the idea, we ignore the stochastic input of the system, i.e., we set $\eps = 0$. We represent the deterministic model as follows:
\begin{small}
    \begin{equation}\label{E: sys2}
    \frac{d}{dt}\begin{bmatrix}
    x_1(t)\\ x_2(t) \\ \theta(t) \end{bmatrix}=  \begin{bmatrix}
     \cos{\theta}(t)  & 0\\
    \sin{\theta}(t) & 0\\ 0 & 1
    \end{bmatrix}\begin{bmatrix}
    v(t)\\
    u(t)
    \end{bmatrix} 
\end{equation} 
\end{small}

\noindent where $X=(x_1, x_2)\in\R^2$ represents the position, $\theta\in[-\pi, \pi]$ is the angle between the $x_1$ axis and the velocity vector   of
the attacker.  The control inputs are $u=(v, w)$, where  $v$ is the velocity and $w$ represents the lateral acceleration.  We consider the initial condition as $(x_1(0), x_2(0), \theta(0))=(-2, 0, -\pi/4)$, and let $x_{e,\alpha}=(0, 0)$ and $x_{e,\beta}=(-0.2, 0.4)$, and assign the targets as $\Gamma_j:=\set{x\in\R^2, \|x-x_{e,j}\|\leq 0.1}$ for $j\in\set{\alpha, \beta}$.   To reach each target $\Gamma_j$, $j\in\set{\alpha, \beta}$, we assume that the attacker follows an optimal guidance law with the common objective of minimizing its control effort, defined as $\mathcal{J}_j(u)=\int_0^\infty 10\|x(t)-x_{e,j}\|^2 + \|u(t)\|^2 dt$. The optimal control laws $\kappa_j(x)$ can then be obtained accordingly.

We also set the discrete-time observation sampling period to be $\delta_t = 0.01$, and the observations to be $Y_n = X_n + V_n$, where $X_n = X(n\delta_t)$,  $V_n\sim \mathcal{N}(0, R_n)$,  and $R_n=\operatorname{diag}(0.1, 0.06)$. For each dimension of the observer and each observation instant, we assume the probability of outlier occurrence is $0,02$, i.e., $\theta_{i,l} = 0.98$ for all $i$ and all $l\in\set{1, 2}$. The indicator value for outlier appearance is set to $\varsigma = 0.08$.

\begin{rem}
   For nonlinear systems, a nonlinear Hamilton–Jacobi–Bellman (HJB) equation must be solved to construct the controllers. This nonlinear problem has been extensively studied in the literature \cite{beard1995improving, meng2024physics, jiang2014robust}, so we omit further details here. Alternatively, the optimal strategy can also be approximately obtained by linearizing the system and applying the Linear–Quadratic Regulator. \Qed
\end{rem}

Fig.~(\ref{fig:subfig1}) shows the missile  trajectories under different realizations of the deceptive change instants.  It can be seen that the trajectory tends to be smoother when the switch point occurs earlier. Correspondingly, as shown in Fig. (\ref{fig:subfig2}), when the observations are contaminated by a fair amount of noise and possible outliers, the deception effect becomes more convincing at earlier switch instants from the deception agent’s point of view, as the observer can hardly distinguish the trajectory trend based solely on the value of $Y$ visually. To make the deceptive switching strategy perform well, the agent can set $d = 0.05$, whence $\pi_k= 0.05\times 0.95^{k-1}$. Below, we present two scenarios of target prediction: one in which the prior distribution $\pi_k$ of the deceptive instants  is known, and another in which the distribution  is unknown. In both cases, the agent draws a switching moment from the distribution $\pi_k$, and we only use realization $\nu = 10$ (i.e., $t=0.1$) to demonstrate the numerical results. 

A. \textit{Posterior Target Prediction with Known Prior Knowledge of Switching Moments.} In this scenario, the outlier-robust filter-induced likelihood $\widehat{L}_n$ is first used to estimate the posterior probability of a deceptive switch in the target. We require the largest false alarm probability to be $a= 0.001$, and use the approximation sequence ${\widehat{L}_n}$ to find the stopping time $\tau_s$ that minimizes the ADD. Given the realization in the corresponding scenario shown in Fig. (\ref{fig:subfig2}), the state at the minimizing moment is marked in Fig. (\ref{fig:subfig1}), indicating a reasonably good timing for detecting the deceptive switching behavior and preventing potential hazards from occurring. 

We also test a total of $1000$ realizations of $Y$, and the range of posterior probabilities computed based on $\widehat{L}_n$ is plotted in Fig. (\ref{fig:subfig3}). The quantitative estimation indicates that after time $0.6$ (also marked by a yellow cross in Fig. (\ref{fig:subfig1})), it is almost certain that a deceptive switch has occurred, and some action should be taken by the counter-deception agent.

B. \textit{Target Prediction with Unknown Prior Knowledge of Switching Moments.} 
In this scenario, since the prior knowledge of $\pi_k$ is unknown, we directly use the statistic $\{\widehat{\Lambda}_n\}$ to track the worst-case estimation of the CADD. Similarly to the previous scenario, we mark the state at the minimizing moment in Fig. (\ref{fig:subfig1}), given the realization in the corresponding scenario shown in Fig. (\ref{fig:subfig2}). Although this optimal stopping time occurs later than in the previous case, it still indicates a reasonably good timing for detecting the deceptive switching behavior.  We also test a total of $1000$ realizations of $Y$, and observe the random stopping time when the statistic exceeds the threshold of $4e4$, as plotted in Fig.~(\ref{fig:subfig4}). The range of this random stopping moment is concentrated around time $0.5$, though not uniformly across all sample paths when compared to the final scenario.  This quantitative estimation  performs well in informing when the observer  should take action. 

\section{Conclusion}
In this paper, we discuss an alternative formulation of target prediction, where the agent is initially believed to be aiming at one target but decides to switch to another midway. The agent behaves deceptively, taking advantage of the fact that the observer only has access to noisy observations and can hardly detect the change in signal visually. We contribute by introducing a detection strategy based on the discussed deceptive behavior, which has not been explored in the literature within the context of deception under constraints of limited and imperfect observations. We enhance the robustness of inference by deriving an outlier-robust formulation of the likelihood function, which is subsequently used to estimate the posterior probability of whether a deceptive switch has occurred—thereby improving computational performance in tracking statistics. Moreover, this likelihood function aligns well with the rich literature on QCD algorithms, enabling a reduction in the number of observations required to determine whether the deceptive switching has taken place. The method is tested on a weapon-target assignment   problem and performs well in fulfilling the task.

Although the model used in this paper considers only two targets with a known prior distribution, the framework can be extended to more complex scenarios. Inspired by the rich QCD literature, a natural extension is to address cases where the pre- or post-change control strategy is unknown. Assuming they belong to single-parameter exponential families, a generalized likelihood ratio approach can be used to infer post-change statistics and determine the optimal stopping time from observations simultaneously. We can then extend this approach to multiple-target detection under deceptive switching strategies, using a similar methodology as outlined above. Interesting formulations can be expected and will be rigorously analyzed to understand the effects of target placement and the properties of the families of the agent's control laws, especially when the task goes beyond simple reachability or stability.

\bibliographystyle{plain}
\bibliography{acc24}

\end{document}